\documentclass{article}
\usepackage[T1]{fontenc}
\usepackage[utf8]{inputenc}
\usepackage{lmodern}
\usepackage{dsfont}
\usepackage{amssymb}
\usepackage{amsmath}
\usepackage{amsthm}
\usepackage{amsfonts, float}
\usepackage{mathabx}
\usepackage{lmodern}
\usepackage{enumitem}
\usepackage{hyperref}
\usepackage{graphicx}
\usepackage{tabularx}
\usepackage{caption}
\usepackage{subcaption}
\usepackage[all]{xy}
\usepackage{etoolbox}
\usepackage{multirow}
\usepackage{authblk}
\usepackage{comment}

\usepackage{titlesec}

\usepackage{tikz}
\usetikzlibrary{quantikz2}

\usepackage[backend=biber, style=phys, biblabel=brakets, url=true]{biblatex}
\renewcommand{\thesection}{\Roman{section}}
\renewcommand{\thesubsection}{\Alph{subsection}}
\renewcommand{\thesubsubsection}{\arabic{subsubsection}}

\titleformat{\section}{\large\scshape}{\thesection.}{1em}{}
\titleformat{\subsection}{\normalsize\bfseries}{\thesubsection.}{1em}{}
\titleformat{\subsubsection}{\normalsize\bfseries}{\thesubsection.\thesubsubsection.}{1em}{}

\usepackage{hyperref}

\newcommand{\X}{\mathtt{X}} 
\newcommand{\Y}{\mathtt{Y}} 
\newcommand{\Z}{\mathtt{Z}}  
\newcommand{\I}{\mathtt{I}}

\DeclareMathOperator{\lcm}{lcm}

\newcommand{\HH}{\mathtt{H}}
\newcommand{\hh}{\mathcal{H}} 

\newcommand{\CX}{\mathtt{CX}} 
\newcommand{\CZ}{\mathtt{CZ}} 

\newcommand{\RY}{\mathtt{R_Y}}
\newcommand{\RZ}{\mathtt{R_Z}}
\newcommand{\RZZ}{\mathtt{R_{ZZ}}}

\newcommand{\U}{\mathtt{U}} 
\newcommand{\W}{\mathtt{W}}

\newcommand{\ord}{\operatorname{ord}}

\newcommand{\Ad}{\operatorname{Ad}}

\newcommand{\Span}{\operatorname{span}}

\newcommand{\GF}{\mathbb F}

\newcommand{\G}{\mathcal G}

\newtheorem{Definition}{Definition}

\newtheorem{Proposition}{Proposition}

\newcommand{\g}{\mathfrak{g}}

\hypersetup{
    colorlinks=true,
    linkcolor=blue
    }

\title{Dynamical Lie Algebras for Fixed-Gate Quantum Ans\"atze}
\author[1,2]{Orlane Zang}
\author[1]{Daniel Juteau}
\author[2]{Tony Quertier}
\author[2]{Grégoire Barrué}
\affil[1]{LAMFA, Université de Picardie Jules Verne, CNRS, France}
\affil[2]{Orange Research, France}
\date{}

\begin{document}

\maketitle

\begin{abstract}
Parameterized Quantum Circuits have emerged as a central framework for variational quantum algorithms in the NISQ regime, with major applications in quantum simulation, combinatorial optimization, and, more recently, quantum machine learning (QML). By developing a new approach for computing dynamical Lie algebras associated with quantum circuits combining parametrized gates and non-parametrized entangling gates, we aim to extend the applicability of recent algebraic results to architectures that are actually used in quantum machine learning.
\end{abstract}

\section{Introduction}
\label{sec:introduction}

Variational Quantum Algorithms (VQAs) use Parameterized Quantum Circuits (PQCs) to prepare families of quantum states or unitaries that are optimized through a classical feedback loop \cite{Peruzzo_2014} \cite{Cerezo_2021} \cite{Biamonte_2017}. Their practical behavior depends on the architecture of the parameterized circuit.  Expressivity, controllability, trainability, and the geometry of the loss landscape are all affected by which infinitesimal directions the circuit can generate as its trainable parameters are varied. A standard tool for characterizing the infinitesimal directions associated with a PQC is the Dynamical Lie Algebra (DLA). In fact, for a circuit whose gates are continuously parameterized as $e^{-i\theta_j\HH_j}$, the DLA noted $\g$, is defined as the Lie closure of the skew-Hermitian generators $i\HH_j$, and is given by
\begin{equation}
\g=\left\langle i\HH_1,\ldots,i\HH_K\right\rangle_{\mathrm{Lie}}.
\label{eq:standard-dla}
\end{equation}
It is the smallest real vector space containing the skew-Hermitian generators and closed under commutators. This construction of the DLA is the setting considered in common Lie-algebraic theories of barren plateaus and overparameterization \cite{Larocca_2022}\cite{Larocca_2023}\cite{Fontana_2024}\cite{Ragone_2024}, 
where the parameterized gates in the Ansatz are collectively governed by a global vector of parameters $\boldsymbol{\theta}$ partitioned into multi-parameters $\boldsymbol{\theta}_l$ for each layer. Specifically, the overall unitary and its layer-wise components take the form
 \begin{equation}
\label{eq: parameterized ansatz}
V(\boldsymbol{\theta}) =\prod_{\ell=1}^{L} V_\ell (\boldsymbol{\theta}_\ell) \quad \text{with} \quad V_\ell (\boldsymbol{\theta}_\ell) =  \prod_{k=1}^{K}  e^{-i\theta_{\ell,k} \HH_{\ell,k}},
\end{equation}
where $L$ is the depth of the PQC, $K$ represents the number of parameterized gates in each layer,  and $\boldsymbol{\theta}_{l}=(\theta_{l,1}, \dots, \theta_{l,K})$ denotes the multi-parameter vector associated with the $l$-th layer.

Many practically used Ans\"atze, for instance some Hardware-Efficient Ans\"atze (HEAs), however, can often interleave local trainable rotation gates with fixed entangling gates, such as controlled-$\X (\CX)$ or controlled-$\Z (\CZ)$ gate \cite{iwakiri2025universalityhardwareefficientansatzes} \cite{Kandala_2017} \cite{Cerezo_2021}. Such fixed gates are not associated with trainable continuous parameters and therefore, they should not be inserted directly into the generator set $\{\HH_1, \cdots, \HH_K\}$ used in Eq. \eqref{eq:standard-dla} to compute the DLA. Indeed, we can write an arbitrary fixed gate as
\begin{equation*}
U=e^{-i\mathtt{K}},
\end{equation*}
and one cannot infer that $i \mathtt{K}$ is an available infinitesimal generator unless the architecture actually provides the continuous family $e^{-it \mathtt{K}}$ with tunable $t$.  Treating a fixed gate as if it were continuously parameterized would not give the correct control system.

This issue was recently clarified in the literature, where West \textit{et al.} \cite{West_2024} initially treated fixed $\CZ$ gates as independent contributions to the DLA by writing
\begin{equation*}
\CZ_{i,i+1}=e^{-i\mathrm{cz}_{i,i+1}}, \qquad \mathrm{cz}_{i,i+1} = -\frac{\pi}{4} \left(\I - \Z_i - \Z_{i+1} + \Z_i \Z_{i+1}
\right),
\end{equation*}
and therefore including the corresponding $\mathrm{cz}_{i,i+1}$ among the generators. However, Xiao and Li \cite{xiao2025commentprovablytrainablerotationally} showed through a counterexample that this procedure is not valid for fixed gates, since introducing the generator of an unparameterized gate directly into the Lie closure can lead to a DLA that does not correctly represent the dynamics of the parameterized circuit. Consequently, West \textit{et al.} issued an erratum \cite{g4zk-79l9}, clarifying that their derivation applies to arbitrary-angle controlled-$\Z$ rotations rather than to fixed $\CZ$ gates, thereby confirming that fixed gates should not be treated as independent generators in the computation of the DLA.

On the other hand, fixed gates cannot simply be ignored.  

Suppose we have a quantum state $\ket{\psi} \in \hh = \hh (2^n)$, to which we successively apply two unitary operators $\U$ and $\mathtt{A}$. Recall that a quantum operator $\U$ is unitary if it is a linear operator $\U: \hh \to \hh$, satisfying
\begin{align*}
\U \U^\dagger = \U^\dagger \U = \I,
\end{align*}
where $\U^\dagger$ is the Hermitian conjugate of $\U$ and $\I$ is the identity operator. Then,
\begin{align*}
\mathtt{A} \U \ket{\psi} =  \U \big(\U^\dagger \mathtt{A} \U \ket{\psi} \big),
\end{align*}
meaning that applying $\U$ to the original state $\ket{\psi}$, and then applying $\mathtt{A}$, is equivalent to first applying the transformed operator $\U^\dagger \mathtt{A} \U$ to the original state, and then applying  $\U$. We can say that the operator $\mathtt{A}$ evolves under $\U$ such as:
\begin{align}\label{eq:conj}
\mathtt{A} \to \U^\dagger \mathtt{A} \U.
\end{align}

If $\mathtt{A}$ is a trainable one-parameter gate taking the form
\begin{equation*}
e^{-i\theta \HH},
\end{equation*}
its evolution under $\U$ will be
\begin{equation*}
e^{-i\theta \HH} \to \U^\dagger e^{-i\theta \HH} \U,    
\end{equation*}
and we will have the following equivalence: 
\begin{equation}\label{eq:conj1}
e^{-i\theta \HH} \ \U \ket{\psi}=  \U \big(\U^\dagger e^{-i\theta \HH} \U \big)\ket{\psi}.    
\end{equation}
Using the fact that for a square real or complex matrix $A$ and for a non-singular matrix $B$,
\begin{align*}
e^{B^{-1} A B}  = B^{-1} e^A B,
\end{align*}
from Eq. \eqref{eq:conj1} we finally have:
\begin{equation}
e^{-i\theta \HH} \ \U = \U e^{-i\theta (\U^\dagger \HH\U )} = \U e^{-i\theta\Ad_{\U}(\HH)},
\label{eq:intro-conjugation}
\end{equation}
where 
\begin{align}
\Ad_{\U}(\HH) = \U^\dagger  \HH \U.  
\end{align}
Thus, for a circuit composed of repeated layers of the form $\U e^{-i\theta \HH}$, the first layer involves the original trainable generator $\HH$. Starting from the second layer, the fixed block $\U$ from the preceding layer transforms the trainable generator by conjugation, so that $\HH$ becomes $\U^\dagger \HH \U$. Repeating the same architecture then generates the sequence
\begin{equation*}
\HH,\qquad \Ad_{\U}(\HH),\qquad \Ad_{\U}^{2}(\HH),\qquad\ldots.
\end{equation*}
Therefore it makes sense to define the DLA in this case as the Lie algebra generated by the successive $\Ad_\U$-conjugates of the original generators (which correspond to the first layer).

This paper formalizes that construction. 
We consider homogeneous fixed-block Ans\"atze, namely architectures in which the same single fixed unitary block $\U$ is repeated at every layer. After rewriting the circuit in terms of effective generators obtained by successive adjoint actions, we then focus on Clifford fixed blocks and Pauli-string generators. In this setting, conjugation induces a finite permutation of Pauli srings, which implies that the set of effective generators obtained after conjugation is finite. The corresponding finite orbit gives a upper bound on the depth at which we have sufficiently many layers to generate the full DLA.

\paragraph{Main contributions.} The main contributions of this work can be summarized as follows:
\begin{enumerate}[leftmargin=1.7em]
\item We define the Dynamical Lie Algebra for QPCs where parametrized layers are alterned with identical non-parametrized layers with unitary $\U$. The idea is that we push all non-parametrized gates to the right, at the cost of conjugating the Lie generators: hence the iterated $\Ad_U$-conjugates of the initial generators must be included.

\item When the fixed block $\U$ is Clifford and the generators to conjugate consist of Pauli strings, we bound the Lie saturation depth by the order of a binary symplectic matrix over $\GF_2$.

\item We specialize the framework to an example of an Ansatz, consisting of $\RY$ rotations and $\CX$ gates. We derive its binary matrix representation, and obtain explicit conjugation periods.
\end{enumerate}

\section{Effective generators}

\subsection{Circuit model}

Let $\mathcal H = (\mathbb C^2)^{\otimes n}\simeq \mathbb C^{2^n}$, and consider a depth $L$-Ansatz in which each layer contains a continuously parameterized block, followed, in physical time, by a fixed unitary block $\U$.  With the standard convention that operators acting first appear on the right, the $L$-layer unitary (Fig. \ref{fig:HEA}) is written
\begin{equation}
V_L(\boldsymbol\theta) = \Big(\U \ \W_L(\boldsymbol{\theta}_L) \Big)\cdots \Big(\U \ \W_2(\boldsymbol{\theta}_2) \Big) \Big(\U \ \W_1(\boldsymbol{\theta}_1) \Big).
\label{eq:circuit-general}
\end{equation}

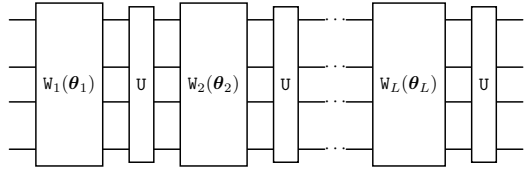
\begin{figure}[htpb]
\centering
\scalebox{0.7}{
\begin{quantikz}
 & \gate[4]{\mathtt{W}_1(\boldsymbol{\theta}_1)} & \gate[4]{\U} & \gate[4]{\mathtt{W}_2(\boldsymbol{\theta}_2)} & \gate[4]{\U} & \cdots & \gate[4]{\mathtt{W}_L(\boldsymbol{\theta}_L)} & \gate[4]{\U} &   \\ 
 & & & & & \cdots & & &  \\
&\cdots &\cdots &\cdots &\cdots & \cdots & & &  \\
 & & & & & \cdots & & &   
\end{quantikz}}
\caption{\textbf{General structure of the Hardware-Efficient Ansatz (HEA) considered in this work.} The circuit consists of $L$ successive layers, each combining a parameterized block $\mathtt{W}_l(\boldsymbol{\theta}_l)$, where $\boldsymbol{\theta}_l$ denotes the trainable parameters of layer $l$, followed by a fixed, unparameterized unitary block $\U$ that is identical across all layers.}
\label{fig:HEA}
\end{figure}

 Each trainable block has the form
\begin{equation}
\W_\ell(\boldsymbol\theta_\ell) = \prod_{k=1}^{K} e^{-i\theta_{\ell,k}\HH_k},
\label{eq:Wblock}
\end{equation}
where the hermitian operators
\begin{equation}
\G_0=\{\HH_1,\ldots,\HH_K\}
\label{eq:G0}
\end{equation}
are the continuously available generators of one layer.

Three assumptions underlying Eq.~\eqref{eq:circuit-general} should be stated explicitly: 

(i) the fixed block $\U$ is assumed to belong to the Clifford group \cite{Garc_a_lvarez_2020} \cite{Bittel_2026} and to be identical in every layer; 

(ii) the continuously parameterized block is repeated across layers with the same generator set $\G_0$, while the corresponding parameter values are independent from one layer to another;

(iii) the fixed block $\U$ may itself be decomposed as
\begin{equation*}
\U = \U_q \U_{q-1}\cdots \U_1,
\end{equation*}
where each $\U_j$ denotes a fixed gate.

\subsection{Effective generators and dynamical Lie algebra}

For a trainable gate $e^{-i\theta \HH},$ the derivative at $\theta=0$ gives
\begin{equation*}
\left.\frac{d}{d\theta}\big(e^{-i\theta \HH} \big)\right|_{\theta=0}=-i \HH,
\end{equation*}
which is precisely the infinitesimal direction entering the DLA.

When we have fixed gates, they can all be gathered to the left of the matrix product (i.e., to the right of the circuit), using the transformation introduced in Eq.~\eqref{eq:intro-conjugation}. In this convention, a trainable gate is transformed only when it is moved across a fixed gate. For a single layer, no rearrangement is required:
\begin{align}
V_1 &= \U \prod_{k=1}^{K} e^{-i\theta_{1,k}\HH_k}.
\end{align}
Hence, at this depth, the available trainable generators are simply the original generators $\HH_k$.
For two layers,
\begin{align}
V_2 &= \U \prod_{k=1}^{K} e^{-i\theta_{2,k}\HH_k} \U \prod_{k=1}^{K}
e^{-i\theta_{1,k}\HH_k}.
\label{eq:v2}
\end{align}
Using Eq.~\eqref{eq:intro-conjugation}, we move the trainable block of the second layer toward the input side of the circuit (right side in Eq. \eqref{eq:v2}), through the fixed block $\U$ of the first layer. This gives
\begin{align}
V_2 &= \U^2 \left(\prod_{k=1}^{K}
e^{-i\theta_{2,k}\Ad_{\U}(\HH_k)} \right) \left(\prod_{k=1}^{K} e^{-i\theta_{1,k}\HH_k} \right).
\end{align}
Thus, the generators of the first trainable block remain unchanged, whereas those of the second block are transformed as
\begin{align*}
\HH_k \longrightarrow \Ad_{\U}(\HH_k), \ \ \text{where} \ \ \Ad_{\U}(\HH) = \U^\dagger  \HH \U
\end{align*}
Similarly,
\begin{align}
V_3 &= \U^3 \left(\prod_{k=1}^{K} e^{-i\theta_{3,k}\Ad_{\U}^{2}(\HH_k)} \right) \left(\prod_{k=1}^{K} e^{-i\theta_{2,k}\Ad_{\U}(\HH_k)} \right)
\left(\prod_{k=1}^{K} e^{-i\theta_{1,k}\HH_k}
\right),
\end{align}
so that the first trainable block therefore contains the original generators $\HH_k$, the second contains $\Ad_{\U}(\HH_k)$, and the third contains $\Ad_{\U}^{2}(\HH_k)$.

More generally, repeating the same rearrangement for a circuit of depth $L$ yields
\begin{equation}
V_L(\boldsymbol{\theta}) = \U^L \prod_{\ell=L}^{1} \left[\prod_{k=1}^{K} e^{-i\theta_{\ell,k}\Ad_{\U}^{\ell-1}(\HH_k)} \right],
\label{eq:VL-rewrite}
\end{equation}
where the outer product is ordered from $\ell=L$ to $\ell=1$, consistently with Eq. \eqref{eq:circuit-general}.

Therefore, the trainable directions appearing throughout the circuit are generated from the original set $\{\HH_k\}$ through repeated applications of the adjoint action,
\begin{equation*}
\HH_k,\qquad
\Ad_{\U}(\HH_k),\qquad
\Ad_{\U}^{2}(\HH_k),\qquad
\ldots,\qquad
\Ad_{\U}^{L-1}(\HH_k),
\end{equation*}
for $1 \leq k \leq K$.
This leads to the following definition.

\begin{Definition}
\label{def:effective-set}
Let $\G_0=\{\HH_1,\ldots,\HH_K\}$ be the set of generators associated with the parametrized gates in a single trainable block, and let $\U$ denote the fixed entangling block. The trainable directions appearing up to depth $L$ are
\begin{equation}\label{eq:effective-set}
\G_L = \bigcup_{j=0}^{L-1}
\Ad_{\U}^{j}(\G_0),
\end{equation}
where
\begin{equation*}
\Ad_{\U}(A) := \U^\dagger A \U.
\end{equation*}
The corresponding depth-$L$ dynamical Lie algebra is then
\begin{equation}
\g_L = \left\langle i\,\G_L \right\rangle_{\mathrm{Lie}}.
\label{eq:depth-dependent-dla}
\end{equation}
\end{Definition}

Definition \ref{def:effective-set} can be applied to a known circuit such as the one shown in Fig. \ref{fig:RZZ decom}, which depicts the $\mathtt{R_{ZZ}}(\theta)$ gate (Fig. \ref{fig:RZZ}) and its decomposition into a sequence of $\CX$ and $\RZ(\theta)$ gates (Fig. \ref{fig: RZZ decomposition}) \cite{Whitfield_2011} \cite{Mansky2023bsb}.
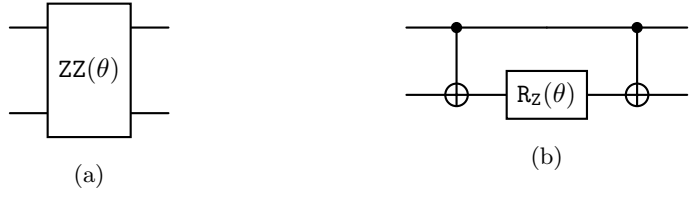
\begin{figure}[H]
    \centering
    \begin{subfigure}[t]{0.48\textwidth}
        \centering
  \begin{quantikz}
   & \gate[2]{\Z \Z(\theta)} &  \\
  &  & 
  \end{quantikz}
        \caption{}
        \label{fig:RZZ}
    \end{subfigure}%
    ~ 
    \begin{subfigure}[t]{0.48\textwidth}
        \centering
 \begin{quantikz}
   & \ctrl{1} & & \ctrl{1} & \\
  & \targ{} & \gate{\RZ (\theta)} &  \targ{} &
  \end{quantikz}
        \caption{}
    \label{fig: RZZ decomposition}
\end{subfigure}
\caption{\textbf{The $\RZZ(\theta)$ gate and its decomposition into elementary gates.} Figure (a) represents the two-qubit $\RZZ(\theta)$ rotation, while Figure (b) shows its decomposition into two CNOT $\CX$ gates surrounding a single-qubit rotation $\RZ(\theta)$, where $\theta$ denotes the rotation angle.}
    \label{fig:RZZ decom}
\end{figure}
For the original circuit in Fig.~\ref{fig:RZZ}, the Lie-algebraic structure is immediate, since the parameterized gate is generated by
\begin{equation*}
\G={\Z\otimes\Z},
\end{equation*}
and therefore
\begin{equation*}
\g = \Span\{i\Z\otimes\Z\}.
\end{equation*}

We now verify that the same generator is recovered from the decomposed circuit in Fig.~\ref{fig: RZZ decomposition}. The parameterized $\RZ$ gate is conjugated by the surrounding $\CX$,
\begin{equation*}
\U^\dagger (\I\otimes\RZ)\U, \qquad \U=\U^\dagger=\CX_{1,2}.
\end{equation*}
Since $\RZ(\theta)=e^{-i\frac{\theta}{2}\Z}$, and knowing that 
\begin{align*}
&\CX_{1,2} (\Z \otimes \I ) \CX_{1,2} = \Z \otimes \I, && \CX_{1,2} (\I \otimes \Z) \CX_{1,2} = \Z \otimes \Z,
\end{align*}
we get:
\begin{align*}
\U^\dagger \left(\I\otimes e^{-i\frac{\theta}{2}\Z}\right)\U &= e^{-i\frac{\theta}{2}\U^\dagger(\I\otimes\Z)\U} \\ &= e^{-i\frac{\theta}{2}\Z\otimes\Z}.
\end{align*}
Hence, the effective generator of the decomposed circuit is again $\Z\otimes\Z$, yielding
\begin{equation*}
\g = \Span\{i\Z\otimes\Z\}.
\end{equation*}
Thus, the decomposition preserves the generator and, consequently, the dynamical Lie algebra of the original $\RZZ$ gate.

\section{Depth-dependent growth of the effective generator set}
\label{sec:generator-growth}

\subsection{Depth-dependent growth of the dynamical Lie algebra}
\label{sec:dla-growth}

For parameterized circuits in which the generator set does not change from one layer to another, the DLA is directly given by
\begin{equation*}
\g = \left\langle i\G_0\right\rangle_{\mathrm{Lie}},
\end{equation*}
where $\G_0=\{\HH_1,\ldots,\HH_K\}$ is the set of the generators in one layer \cite{Larocca_2023} \cite{Ragone_2024}. In our setting, the fixed unitary block $\U$ will change the effective generator
from one layer to the next, i.e,
\begin{equation*}
\G_L = \bigcup_{j=0}^{L-1} \Ad_{\U}^{j}(\G_0),
\end{equation*}
so that
\begin{equation}
\G_{L+1} = \G_L \cup \Ad_{\U}^{L}(\G_0).
\label{eq:GLplus1}
\end{equation}
Hence
\begin{equation}
\G_L \subseteq  \G_{L+1},
\end{equation}
and the corresponding depth-dependent Lie closures satisfy
\begin{equation}
\g_L \subseteq \g_{L+1},
\label{eq:dla-monotonic-growth}
\end{equation}
implying
\begin{equation}
\dim(\g_L) \leq \dim(\g_{L+1}).
\end{equation}

At the end, we therefore obtain a filtration of the DLA $\g$ by Lie subalgebras, such that:
\begin{equation}
\g_1 \subseteq \g_2 \subseteq \cdots \subseteq \g,
\label{eq:dla-filtration}
\end{equation}
where
\begin{equation}
\g = \left\langle i \bigcup_{j=0}^{\infty} \Ad_{\U}^{j}(\G_0) \right\rangle_{\mathrm{Lie}}
\end{equation}
denotes the DLA of the circuit. This sequence is of course stationary because the dimension of $\mathfrak g \subset \mathfrak{su}(2^n)$ is finite.
We want to determine the minimal $L$ such that $\mathfrak g_L = \mathfrak g$, which we will call the \emph{Lie saturation depth}, denoted by $L_{\mathrm{Lie}}$. 

To be clear, $L_{\mathrm{Lie}}$ is not the overparametrization threshold. A priori, more layers are needed to actually fully explore all possible directions. A lower bound for this threshold is given by the obvious necessary condition that the total number of parameters must be greater than the dimension of the DLA, i.e. $KL \geq \dim \mathfrak g$ \cite{Larocca_2023}.

To address the question of determining $L_{\mathrm{Lie}}$, we examine the sequence of effective generator sets:
\begin{equation*}
\G_0,\qquad \Ad_{\U}(\G_0),\qquad \Ad_{\U}^{2}(\G_0),\qquad \ldots
\end{equation*}

Since $\Ad_{\U}$ permutes Pauli operators up to phase, the $\Ad_{\U}$-orbit of each generator (up to phase) is finite; thus the length of the $\Ad_\U$-orbit of $\G_0$ is an upper bound for $L_{\mathrm{Lie}}$. We will see in the next section that $\Ad_\U$-conjugation on Pauli strings up to phase is described by a binary matrix $M_{\U}$, hence the length of the orbit divides the order $T_{\U}$ of $M_{\U}$. We will see that this order can be rather large (depending on the number of qubits), but $L_{\mathrm{Lie}}$ may be much smaller, if the first few conjugates are enough to generate all others by taking Lie brackets.

\subsection{Periodicity of the generators under conjugation}
\label{sec:finite-recurrence}

The fixed block in Eq. \eqref{eq:circuit-general} is assumed to be a Clifford operator. We recall that the defining property of Clifford operators is that they normalize the $n$-qubit Pauli group, i.e conjugating a Pauli operator by a Clifford operator produces another Pauli operator, up to a phase
\cite{gottesman1998heisenbergrepresentationquantumcomputers} \cite{Dehaene_2003}. The generators in $\G_0$ being Pauli strings, every element generated by repeated applications of $\Ad_{\U}$ remains a Pauli string up to phase. Since the phase is irrelevant to compute the generated Lie algebra, we will work modulo scalars, and use [ ] brackets to indicate that.

The set of Pauli strings modulo phase being finite, repeated conjugation by the same Clifford operator must eventually lead to a repetition,
immediately implying the following Proposition. 

\begin{Proposition}
\label{prop:common-clifford-recurrence}
Let $\G_0$ be a set of Hermitian $n$-qubit Pauli strings,
and let $\U$ be a fixed Clifford operator.
Then, there exists a finite integer $d>0$ such that
\begin{equation}
\Ad_{\U}^{d}([\G_0])=[\G_0].
\label{eq:common-clifford-period}
\end{equation}
Consequently, no new effective generator directions
appear beyond depth $d$, and
\begin{equation}
\g_L=\g_d, \qquad \forall L\geq d.
\end{equation}
\end{Proposition}

Proposition~\ref{prop:common-clifford-recurrence} shows that once all elements of the corresponding conjugation orbits have appeared, increasing the circuit depth does not generate any new effective directions.

The next question is therefore to determine $d$. However, following the repeated conjugations directly becomes more difficult as the number of qubits increases.

The Clifford formalism provides a convenient way to study this recurrence. Indeed, Pauli operators modulo phase can be encoded as binary vectors, and Clifford conjugation acts on these vectors through an invertible linear transformation over $\GF_2$ \cite{Dehaene_2003}. The problem can therefore
be reduced to studying successive powers of the associated binary matrix.

\section{Binary Clifford Representation for Determining the Period \texorpdfstring{$d$}{d}}

In this section, we use the standard stabilizer formalism and the binary representation of Clifford operations \cite{gottesman1998heisenbergrepresentationquantumcomputers} \cite{Dehaene_2003} \cite{Hostens_2005} \cite{Helsen_2018} to determine the period $d$. We first represent Pauli strings by vectors in the finite field $\GF_2^{2n}$, and then study the corresponding action of the Clifford operator $\U$ in this representation. This allows us to analyze the repeated conjugations through the powers of the associated binary matrix.

\subsection{Binary representation of Clifford conjugation}
\label{sec:methods-clifford}

\subsubsection{Binary representation of the Pauli group}

The group generated by tensor products of the four single-qubit Pauli operators $\I$, $\X$, $\Y$, and $\Z$ is called the $n$-qubit Pauli group, given by
\begin{equation}
\mathcal{P}_n = \{\pm 1,\pm i\} \cdot \{\I,\X,\Y,\Z\}^{\otimes n},
\end{equation}
and contains $4^{n+1}$ elements. Any two elements of $\mathcal{P}_n$ either commute or anticommute, and the square of every element is equal to $\pm\I$.

In what follows, since we work with a binary representation over $\GF_2$, which encodes only the Pauli operators and not their global phases, we can therefore consider the quotient group
\begin{equation}
[\mathcal{P}]_n = \mathcal{P}_n/Z(\mathcal{P}_n) = \{[\I],[\X],[\Y],[\Z]\}^{\otimes n},
\end{equation}
where
\begin{equation*}
Z(\mathcal{P}_n) = \{\pm\I^{\otimes n},\pm i\I^{\otimes n}\}
\end{equation*}
is the center of $\mathcal{P}_n$, and $[\I],[\X],[\Y]$ and $[\Z]$ respectively denote the equivalence classes of $\I$, $\X$, $\Y$, and $\Z$ modulo the global phases $\{\pm1,\pm i\}$.  Every Pauli operator in $[\mathcal{P}]_n$ can be written as
\begin{align*}
\begin{aligned}
[P](x;z) &= \left([\X]^{x_1}\otimes\cdots\otimes[\X]^{x_n}\right)
\left([\Z]^{z_1}\otimes\cdots\otimes[\Z]^{z_n}\right) \\
&= [\X]_1^{x_1}\cdots[\X]_n^{x_n}
[\Z]_1^{z_1}\cdots[\Z]_n^{z_n}
= [\X]^{x}[\Z]^{z}.
\end{aligned}
\end{align*}
with
\begin{equation*}
x=(x_1,\ldots,x_n) \in\GF_2^n,\qquad z=(z_1,\ldots,z_n) \in\GF_2^n, 
\end{equation*}
where we use the row-vector convention adopted in Ref. \cite{Hashagen_2018}. The Pauli string $[P](x;z)$ is then represented by the binary vector $(x;z)\in\GF_2^{2n}$. In particular,
\begin{equation}
[\I]\longleftrightarrow(0;0),\qquad
[\X]\longleftrightarrow(1;0),\qquad
[\Z]\longleftrightarrow(0;1),\qquad
[\Y]\longleftrightarrow(1;1),
\end{equation}
where the last correspondence follows from $[\Y] =[\X] [\Z]$. Thus, the projective Pauli group $[\mathcal{P}]_n$ is isomorphic to the binary vector space $\GF_2^{2n}$, i.e
\begin{align*}
[\mathcal{P}]_n\cong \GF_2^{2n}.
\end{align*}

\subsubsection{Binary representation of Clifford conjugation}

The $n$-qubit Clifford group, which is the normalizer of the Pauli group in the unitary group $\mathcal{U}$, is defined as:
\begin{equation}
\mathrm{Cliff}_n = N_{\mathcal{U}(2^n)}(\mathcal{P}_n) = \left\{\mathtt{C}\in \mathcal{U}(2^n): \mathtt{C}^\dagger\mathcal{P}_n\mathtt{C} = \mathcal{P}_n \right\},
\end{equation}
i.e,
\begin{equation}
\mathtt{C}^\dagger P\mathtt{C} \in \mathcal{P}_n, \qquad \forall\,P\in\mathcal{P}_n.
\end{equation}

Since conjugation by $\mathtt{C}$ maps Pauli operators to Pauli operators and is invertible, it permutes the elements of $\mathcal{P}_n$. Moreover, conjugation preserves the center $Z(\mathcal{P}_n)$, so this action induces a well-defined permutation on the quotient $[\mathcal{P}]_n=\mathcal{P}_n/Z(\mathcal{P}_n)$. Through the identification $[\mathcal{P}]_n\cong \GF_2^{2n}$, the induced action can be represented directly on the associated binary vectors by a matrix in $\mathrm{GL}(2n,\GF_2)$.

For two Pauli classes represented by
\begin{equation*}
v=(x;z), \qquad v'=(x';z'),
\end{equation*}
their product satisfies
\begin{equation*}
[P](v)[P](v') = [P](v+v').
\end{equation*}
Since Clifford conjugation preserves this product and is invertible, it induces an
invertible linear transformation on $\GF_2^{2n}$. Therefore, there exists a linear matrix
\begin{equation*}
M_{\mathtt C}\in\mathrm{GL}(2n,\GF_2)
\end{equation*}
such that, using the row-vector convention, we have:
\begin{equation}
\mathtt C^\dagger [P](v)\mathtt C = [P](vM_{\mathtt C}).
\end{equation}

Note that we have been using row vectors and matrices acting on the right, because it must reflect right conjugation, which is a right action.

The commutation relation between two Pauli classes is determined by the binary
symplectic product
\begin{equation}
[v,v'] = vJv'^{T}, \qquad J=
\begin{pmatrix}
0 & \I_n\\
\I_n & 0
\end{pmatrix},
\label{eq:binary-symplectic-form}
\end{equation}
where $[v,v']=0$ for commuting Pauli classes and $[v,v']=1$ for anticommuting ones, and $J$ is the matrix representing this sympletic form over $\GF_2$. Note that $-1=1$ over $\GF_2$. More generally, the matrices preserving this product form the symplectic group
\begin{equation}
\operatorname{Sp}(2n,\GF_2) = \left\{M\in\mathrm{GL}(2n,\GF_2) \,:\, MJM^{T}=J \right\}.
\end{equation}

Since Clifford conjugation preserves commutation relations,
\begin{equation*}
[vM_{\mathtt C},v'M_{\mathtt C}] = [v,v']
\end{equation*}
for all $v,v'\in\GF_2^{2n}$, which implies
\begin{equation*}
M_{\mathtt C}JM_{\mathtt C}^{T} = J.
\end{equation*}
Hence,
\begin{equation*}
M_{\mathtt C}\in\operatorname{Sp}(2n,\GF_2).
\end{equation*}

Repeated conjugations by the same Clifford operator are consequently represented by successive powers of the same symplectic matrix. Indeed, after one conjugation,
\begin{equation*}
[P](x,z) \longmapsto [P]\!\left((x;z)M_{\mathtt{C}}\right),
\end{equation*}
while after two conjugations,
\begin{equation*}
[P](x,z) \longmapsto [P]\!\left((x;z)M_{\mathtt{C}}^{2}\right).
\end{equation*}
More generally, after $k$ successive conjugations by $\mathtt{C}$,
\begin{equation}
\Ad_{\mathtt{C}}^{k}\!\left([P](x,z)\right) = [P]\!\left((x;z)M_{\mathtt{C}}^{k} \right).
\label{eq:repeated-binary-action}
\end{equation}

\subsection{Conjugation period and order of the binary action}
\label{sec:conjugation-order}
We now characterize the recurrence establishes in Proposition \ref{prop:common-clifford-recurrence} using the binary representation introduced above.

From Eq.~\eqref{eq:repeated-binary-action}, after $k$ conjugations by a fixed Clifford operator $\U$, the binary vector associated with a Pauli string evolves as
\begin{equation*}
(x;z) \longmapsto (x;z)M_{\U}^{k}.
\end{equation*}
A complete recurrence of the conjugation action is obtained when every binary Pauli vector returns to its initial value, namely
\begin{equation*}
(x;z)M_{\U}^{k} = (x;z), \qquad \forall\,(x;z)\in\GF_2^{2n}.
\end{equation*}
This is equivalent to
\begin{equation}
M_{\U}^{k}=I_{2n}.
\end{equation}

The smallest positive integer satisfying this relation is the order of $M_{\U}$,
\begin{equation} 
T_{\U} := \ord(M_{\U}) = \min\left\{k\geq1: M_{\U}^{k}=I_{2n} \right\}.
\end{equation}

\begin{Definition}
\label{def:conjugation-order}
Let $\U$ be a fixed Clifford operator and let $M_{\U}$ be the binary matrix representing its conjugation action. The global period of the conjugation action on the projective Pauli group is defined by:
\begin{equation}
T_{\U}=\ord(M_{\U}).
\end{equation}
Consequently,
\begin{equation}
\Ad_{\U}^{T_{\U}}[\HH_j] = [\HH_j], \qquad \forall\,\HH_j\in\G_0.
\end{equation}
\end{Definition}

The order $T_{\U}$ characterizes the periodicity of the $\Ad_U$ on all Pauli strings. The period for the particular set $\G_0$ may be smaller. In any case, the minimal common period $d$ introduced in Proposition~\ref{prop:common-clifford-recurrence} satisfies
\begin{equation}
d\mid T_{\U}.
\end{equation}

We now examine the binary transformations associated with two elementary Clifford entangling gates that commonly appear as fixed gates in variational quantum circuits, the controlled-$\X$ and controlled-$\Z$ gates.

\subsubsection{Controlled-\texorpdfstring{$\X$}{X} gate}

Consider a single controlled-$\X$ gate $\CX_{c,t}$, with control qubit $c$ and target qubit $t$. Its conjugation action on the Pauli generators is
\begin{align}
\begin{aligned}
\Ad_{\CX_{c,t}}(\X_c) &= \X_c\X_t, & \Ad_{\CX_{c,t}}(\X_t) &= \X_t, \\ \Ad_{\CX_{c,t}}(\Z_c) &= \Z_c, & \Ad_{\CX_{c,t}}(\Z_t) &= \Z_c\Z_t,
\end{aligned}
\label{eq:CX-conjugation-rules}
\end{align}
while Pauli operators acting on the remaining qubits are unchanged.

Thus, the $\X$- and $\Z$-components remain separated under conjugation. For the binary vector $(x;z)$, the transformation can be written as
\begin{equation}
x \longmapsto xB_{c,t}, \qquad z \longmapsto zB_{c,t}^{-T},
\end{equation}
where
\begin{equation}
B_{c,t} = I_n+E_{c,t},
\end{equation}
and $E_{c,t}$ denotes the $n\times n$ elementary matrix containing a single nonzero entry in position $(c,t)$.

The complete binary matrix associated with the conjugation action of the controlled-$\X$ gate is therefore
\begin{equation}
M_{\CX_{c,t}} =
\begin{pmatrix}
B_{c,t} & 0\\ 0 & B_{c,t}^{-T}
\end{pmatrix}.
\end{equation}

This structure extends to a fixed block $\U$ composed only of controlled-$\X$ gates. Indeed, since each controlled-$\X$ gate acts separately on the $\X$- and $\Z$-components, the binary matrix associated with the complete block retains the same block-diagonal form,
\begin{equation}
M_{\U} =
\begin{pmatrix}
B_{\U} & 0\\ 0 & B_{\U}^{-T}
\end{pmatrix},
\end{equation}
where $B_{\U}\in\mathrm{GL}(n,\GF_2)$ is obtained by composing the elementary matrices associated with the controlled-$\X$ gates forming $\U$.

After $k$ repetitions of the same fixed block,
\begin{equation*}
M_{\U}^{k} =
\begin{pmatrix}
B_{\U}^{k} & 0\\ 0 & \left(B_{\U}^{-T}\right)^{k}
\end{pmatrix}
=
\begin{pmatrix}
B_{\U}^{k} & 0\\ 0 & \left(B_{\U}^{k}\right)^{-T}
\end{pmatrix}.
\end{equation*}
Therefore,
\begin{equation*}
M_{\U}^{k}= \I_{2n}
\quad\Longleftrightarrow\quad B_{\U}^{k}= \I_n.
\label{eq:CX-order-reduction}
\end{equation*}

Consequently,
\begin{equation}
\ord(M_{\U}) = \ord(B_{\U}).
\label{eq:CX-order-B}
\end{equation}

Hence, for a fixed block composed exclusively of controlled-$\X$ gates, the conjugation period can be determined by considering only one of the two Pauli parts.

\subsubsection{Controlled-\texorpdfstring{$\Z$}{Z} gate}

Consider now a single controlled-$\Z$ gate $\CZ_{c,t}$. Its conjugation rules are
\begin{align}
\begin{aligned}
\Ad_{\CZ_{c,t}}(\X_c) &= \X_c\Z_t, & \Ad_{\CZ_{c,t}}(\X_t) &= \Z_c\X_t, \\ \Ad_{\CZ_{c,t}}(\Z_c) &= \Z_c, & \Ad_{\CZ_{c,t}}(\Z_t) &= \Z_t,
\end{aligned}
\end{align}
while Pauli operators acting on the remaining qubits are unchanged. Introducing
\begin{equation*}
A_{c,t} = E_{c,t}+E_{t,c},
\end{equation*}
the binary action is
\begin{equation*}
(x;z) \longmapsto \left(x;\, z+xA_{c,t} \right).
\end{equation*}

The corresponding binary matrix is therefore
\begin{equation}
M_{\CZ_{c,t}} =
\begin{pmatrix}
\I_n & A_{c,t}\\ 0   & \I_n
\end{pmatrix}.
\label{eq:M-CZ}
\end{equation}

This structure can be extended to a fixed block $\U$ composed only of
controlled-$\Z$ gates. Since, for each controlled-$\Z$ gate, the binary
transformation has the form
\begin{equation}
x\longmapsto x, \qquad z\longmapsto z+xA_{c,t},
\end{equation}
the composition of several such gates preserves the $\X$-component and
modifies only the $\Z$-component through a linear contribution depending on
$x$. Therefore, the binary matrix associated with the complete block retains
the same upper-triangular form,
\begin{equation}
M_{\U} =
\begin{pmatrix}
I_n & A_{\U}\\ 0  & I_n
\end{pmatrix},
\end{equation}
where $A_{\U}\in\GF_2^{n\times n}$ is the sum, over $\GF_2$, of the matrices associated with the controlled-$\Z$ gates forming $\U$.

If
\begin{equation*}
\U = \CZ_{c_q,t_q}\cdots \CZ_{c_2,t_2} \CZ_{c_1,t_1},
\end{equation*}
then, over $\GF_2$,
\begin{equation*}
A_{\U} = \sum_{j=1}^{q} A_{c_j,t_j}.
\end{equation*}
Hence, after $k$ repetitions of the same fixed block,
\begin{equation}
M_{\U}^{k} =
\begin{pmatrix}
\I_n & kA_{\U}\\ 0   & \I_n
\end{pmatrix}.
\end{equation}

Since the calculations are performed over $\GF_2$,
\begin{equation}
M_{\U}^{2} =
\begin{pmatrix}
\I_n & 2A_{\U}\\ 0   & I_n
\end{pmatrix}
= \I_{2n}.
\end{equation}
Therefore,
\begin{equation}
\ord(M_{\U}) =
\begin{cases}
1, & A_{\U}=0,\\[1mm]
2, & A_{\U}\neq 0.
\end{cases}
\end{equation}

Hence, for a fixed block composed exclusively of controlled-$\Z$ gates, the conjugation period is completely determined by the matrix $A_{\U}$. In particular, any nontrivial such block has binary conjugation order two.

\section{An example on a specific Ansatz}

\subsection{Effective generators}

We consider the Ansatz in Fig. \ref{fig:QML model}, composed of $\RY$ rotation gates of the form
\begin{equation}
\mathtt{W}(\boldsymbol{\theta})=\bigotimes _{j=1}^{n} e^{-i\frac{\theta_{j}}{2}\Y_j},
\end{equation}
while the fixed entangling block consists of a ring of $\CX$ gates. 

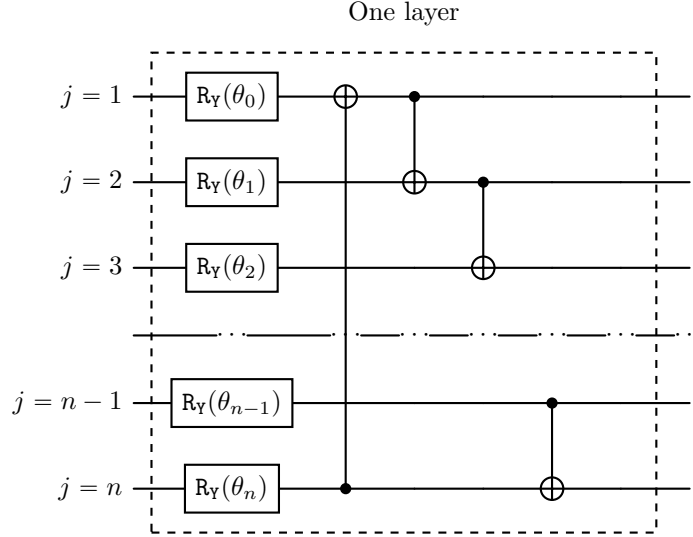
\begin{figure}[htpb]
\centering
\begin{quantikz}
\lstick{$j=1$} & \gate{\RY(\theta_0)} \gategroup[6,steps=6,style={dashed},label style = {yshift=0.3cm}]{One layer} &\targ{} & \ctrl{1} & & & & \\ 
\lstick{$j=2$}  & \gate{\RY(\theta_1)} & & \targ{} & \ctrl{1} & & & \\
\lstick{$j=3$} & \gate{\RY(\theta_2)} & & &\targ{} & & & \\
 &\cdots &\cdots &\cdots  & \cdots &\cdots &\cdots &\cdots  \\
\lstick{$j=n-1$}  & \gate{\RY(\theta_{n-1})} &  & & & \ctrl{1} & & \\
\lstick{$j=n$} &  \gate{\RY(\theta_n)} & \ctrl{-5} & & & \targ{}& & 
\end{quantikz}
\caption{\textbf{One layer of the Ansatz considered in this work.} The layer consists of parameterized $\RY(\theta_j)$ rotations applied to each of the $n$ qubits, followed by a fixed sequence of $CX$ gates arranged in a ring architecture, i.e each qubit acts as the control for the next one, with the last qubit connected back to the first. The dashed box indicates one layer of the circuit, which can be repeated to construct the complete Ansatz.}
\label{fig:QML model}
\end{figure}

In the physical circuit, following the direction of state evolution, these $\CX$ gates are applied in the order
\begin{equation*}
\CX_{n,1},\quad \CX_{1,2},\quad \CX_{2,3},\quad \ldots,\quad \CX_{n-1,n},
\end{equation*}
then the conjugation block $\U$ is 
\begin{equation}
\U = \CX_{n-1, n} \CX_{n-2, n-1}, \cdots \CX_{2,3} \CX_{1,2} \CX_{n,1}.
\end{equation}
Thus, the first conjugation is performed by $\CX_{n-1,n}$ and the last by $\CX_{n,1}$.

The action of the full ring block $\U$ on the three basic Pauli operators is defined by the following relations: 
\begin{subequations}\label{eq:conj XYZ by U}
\begin{align}
 & \left\{ \begin{aligned}
 &\U^\dagger \X_j \U = \X_j\X_{j+1}, \ \text{for} \ j < n-1 \\
 &\U^\dagger \X_{n-1} \U = \X_1\X_{n-1} \X_n, \\ &\U^\dagger \X_n \U = \X_1\X_n,
 \end{aligned} \right.
\end{align}
\begin{align}
 & \left\{ \begin{aligned}
 &\U^\dagger \Z_1 \U = \Z_1\Z_n, \\
 &\U^\dagger \Z_j \U = \Big(\prod_{i = 1} ^{j} \Z_i \Big) \Z_n, \ \text{for} \ 1 < j < n, \\
 &\U^\dagger \Z_n \U = \prod_{i = 1} ^{n-1} \Z_i,
 \end{aligned} \right.
\end{align}
\begin{align}\label{eq:conj Y with U}
 &\left\{ \begin{aligned}
 &\U^\dagger \Y_1 \U = \Y_1 \X_2 \Z_n, \\
 &\U^\dagger \Y_j \U = \Big(\prod_{i = 1} ^{j-1} \Z_i \Big) \Y_j \X_{j+1} \Z_n, \ \text{for} \ 1 < j < n-1, \\
 &\U^\dagger \Y_{n-1} \U = -\Y_1 \Big(\prod_{i = 2} ^{n-2} \Z_i \Big) \Y_{n-1} \Y_n, \\
 &\U^\dagger \Y_n \U = \Y_1 \Big(\prod_{i = 2} ^{n-1} \Z_i \Big) \X_n.
 \end{aligned}
 \right.
\end{align}
\end{subequations}

Then for this specific Ansatz, at depth L, the generators $\HH_j = \Y_j$ at layer $L>1$ will be conjugated according to relations \eqref{eq:conj Y with U}.

\subsection{Binary action of the full entangling block
\texorpdfstring{$\U$}{U}}
\label{subsec:actionU}

We now apply the general construction introduced above to the entangling block
\begin{equation*}
\U = \CX_{n-1,n}\cdots \CX_{2,3}\CX_{1,2}\CX_{n,1}.
\end{equation*}

Since $\U$ is composed exclusively of controlled-$X$ gates, its binary
representation has the block-diagonal form derived previously,
\begin{equation*}
M_{\U} =
\begin{pmatrix}
B_n & 0\\ 0 & B_n^{-T}
\end{pmatrix},
\end{equation*}
where
\begin{equation}
B_n = B_{n-1,n}\cdots B_{2,3}B_{1,2}B_{n,1} \ \ \in \mathrm{GL}(n,\GF_2)
\end{equation}
describes the action on the $\X$-component.

Therefore, for a Pauli operator $P(x,z)=\X^x\Z^z$, with $x,z\in \GF_2^n$ represented as row vectors, the conjugation action $\Ad_{\U}(P) = \U^\dagger P\U$ induces the binary transformation
\begin{equation*}
(x;z) \longmapsto (xB_n\,;\,zB_n^{-T}).
\end{equation*}

After $k$ successive conjugations,
\begin{equation*}
(x;z) \longmapsto (x;z)M_{\U}^{k},
\end{equation*}
with
\begin{equation*}
M_{\U}^{k} =
\begin{pmatrix}
B_n^{k} & 0\\ 0 & (B_n^{k})^{-T}
\end{pmatrix}.
\end{equation*}

Hence, as established in the general controlled-$\X$ case,
\begin{equation*}
\ord(M_{\U})=\ord(B_n),
\end{equation*}
so it is sufficient to determine the order of $B_n$.

\subsection{Derivation of \texorpdfstring{$B_n$}{Bn}}

To determine $B_n$, we consider the action of $\U$ on the Pauli operators $\X_j$. For $n\geq 3$, we have
\begin{equation*}
\begin{aligned}
\U^\dagger \X_j \U &= \X_j\X_{j+1}, &&1\leq j\leq n-2, \\
\U^\dagger \X_{n-1}\U &= \X_1\X_{n-1}\X_n, \\ \U^\dagger \X_n\U &= \X_1\X_n.
\end{aligned}
\label{eq:X-conjugation-ring}
\end{equation*}

Using the binary correspondence $\X_j\leftrightarrow e_j$ and the row-vector convention introduced above, these relations become
\begin{equation*}
\begin{aligned}
e_jB_n &= e_j+e_{j+1}, &&1\leq j\leq n-2, \\ e_{n-1}B_n
&= e_1+e_{n-1}+e_n, \\ e_nB_n &= e_1+e_n.
\end{aligned}
\end{equation*}

Thus,
\begin{equation}
B_n =
\begin{pmatrix}
1 & 1 & 0 & \cdots & 0 & 0 & 0\\
0 & 1 & 1 & \ddots & \vdots & \vdots & \vdots\\
0 & 0 & 1 & \ddots & 0 & 0 & 0\\
\vdots & \ddots & \ddots & \ddots & 1 & 0 & 0\\
0 & \cdots & 0 & 0 & 1 & 1 & 0\\
1 & 0 & \cdots & 0 & 0 & 1 & 1\\
1 & 0 & \cdots & 0 & 0 & 0 & 1
\end{pmatrix}.
\label{eq:Bn}
\end{equation}

The previous derivation assumes $n\geq 3$. The case $n=2$ must be treated separately, since the qubit indexed by $n-1$ coincides with the first qubit. In this case,
\begin{align*}
\U^\dagger \X_1\U &= \X_2,\\ \U^\dagger \X_2\U &= \X_1\X_2.
\end{align*}
Hence,
\begin{equation*}
e_1B_2=e_2, \qquad e_2B_2=e_1+e_2,
\end{equation*}
and therefore
\begin{equation}
B_2 =
\begin{pmatrix}
0&1\\
1&1
\end{pmatrix}.
\end{equation}

\subsection{Order of \texorpdfstring{$B_n$}{Bn}}

The order of a matrix is governed by its minimal polynomial. Recall that the set of polynomials annihilating $B_n$,
\begin{align*}
I_{B_n} = \left\{P\in\GF_2[x] \;:\; P(B_n)=0 \right\},
\end{align*}
is an ideal of $\GF_2[x]$. Since $\GF_2[x]$ is a principal ideal domain, there exists a unique monic polynomial $\mu_{n}$, called the \emph{minimal polynomial} of $B_n$, such that
\begin{align*}
I_{B_n} = (\mu_{n}) = \mu_{n}\GF_2[x].
\end{align*}

Using the relation $\mu_n(B_n)=0$, powers of $B_n$ can be reduced according to the polynomial relation defined by $\mu_n$. Since $B_n\in\mathrm{GL}(n,\GF_2)$, we have
\begin{equation}
B_n^k=I_n \quad\Longleftrightarrow\quad x^k\equiv 1 \pmod{\mu_n}.
\label{eq:order minimal polynomial}
\end{equation}
Thus, $\ord(B_n)$ is entirely determined by $\mu_n$. It is the order of $\overline x$ in $(\GF_2[x] / (\mu_n))^\times$.

In practice, the minimal polynomial can be difficult to determine directly. It is therefore useful to relate it to the characteristic polynomial.

\begin{Proposition}
\label{prop:charact. pol of Bn}
The characteristic polynomial of $B_n$ is
\begin{align}
\chi_n = (x+1)^n+x.
\end{align}
\end{Proposition}

\begin{proof}
The characteristic of $B_n$ is given by (recall that $+1=-1$ in $\mathbb F_2$)
\[
\begin{array}{rcl}
\chi_n &=&
\begin{vmatrix}
x + 1 & 1 & 0 & \cdots & 0 & 0 & 0\\
0 & x + 1 & 1 & \ddots & \vdots & \vdots & \vdots\\
0 & 0 & x + 1 & \ddots & 0 & 0 & 0\\
\vdots & \ddots & \ddots & \ddots & 1 & 0 & 0\\
0 & \cdots & 0 & 0 & x + 1 & 1 & 0\\
1 & 0 & \cdots & 0 & 0 & x + 1 & 1\\
1 & 0 & \cdots & 0 & 0 & 0 & x + 1
\end{vmatrix} \\
\noalign{\vspace{1em}}
&=& (x + 1)\cdot (x + 1)^{n - 1} + 1 \cdot (x + 1) + 1 \cdot 1  \\
&=& (x + 1)^n + x
\end{array}
\]
(We have expanded the determinant along the first column.)
\end{proof}

Let us write the factorization into irreducible polynomials as:
\[
\chi_n = f_1^{m_1} \dots f_s^{m_s}.
\]
We have
\[
f_1 \dots f_s \mid \mu_n \mid f_1^{m_1}\dots f_s^{m_s} = \chi_n.
\]
We will see that in our case, the situation is very simple: all the multiplicities $m_i$ are $1$, implying the $\mu_n = \chi_n$.

\begin{Proposition}
\label{prop:charact. pol = min pol}
The minimal and characteristic polynomials of $B_n$ coincide:
\begin{equation}
\mu_n = \chi_n = (x+1)^n+x = f_1 \dots f_s,
\end{equation}
where the $f_i$ are distinct irreducible polynomials over $\mathbb F_2$.
\end{Proposition}

\begin{proof}
As explained above, it is sufficient to show that $\chi_n$ is square-free, or equivalently (since we work over a finite, hence perfect field):
\begin{align*}
\gcd(\chi_n,\chi_n') = 1.
\end{align*}
We have $\chi_n' = n(x+1)^{n-1}+1$. If $n$ is even then $\chi_n' = 1$ and we are done. So we may assume that $n$ is odd, implying $\chi_n' = (x + 1)^{n - 1} + 1$. We note that
\[
(x + 1) \chi_n' + \chi_n = 1,
\]
implying that $\chi$ and $\chi_n'$ are indeed coprime.
\end{proof}

Consequently, $T_n = \ord(B_n)$ is the order of $\overline x$ in
\[
\left(\GF_2[x]/(\chi_n)\right)^\times
= \left(\GF_2[x]/(f_1\dots f_s) \right)^\times
\simeq \prod_{i=1}^{s} \left(\GF_2[x]/(f_i) \right)^\times
\simeq \prod_{i=1}^{s} \left(\GF_{2^{r_i}} \right)^\times,
\]
where $r_i = \deg f_i$. (We have used the fact that the $f_i$
are distinct irreducible polynomials, and the Chinese Remainder Theorem.)
Under that isomorphism, $\overline x$ corresponds to $(\overline x_i)$,
where $x_i$ is the image of $x$ in $\GF_2[x]/(f_i) = \GF_{2^{r_i}}$. 

We have obtained the following result.

\begin{Proposition}
\label{thm:Order from the irreducible factorization}
Let $t_i$ denote the order of $\overline x$ in $\left(\GF_2[x]/(f_i) \right)^\times \simeq \left(\GF_{2^{r_i}} \right)^\times$.
Then the order $T_n$ of $B_n$ is the lcm of the $t_i$.

Moreover, for each $i$, the order $t_i$ divides $2^{r_i} - 1$, with equality if and only if $f_i$ is primitive (meaning that $\overline x_i$ is a primitive element, i.e. generates $\left(\GF_{2^{r_i}} \right)^\times$).
\end{Proposition}

To sum up, the exact conjugation period $T_n$ can be determined by
factoring
\begin{align*}
\chi_n = (x+1)^n+x
\end{align*}
over $\GF_2$, computing the multiplicative order associated with
each irreducible factor, and taking the least common multiple of these
orders.

\subsection{Specific cases for differents number of qubits \texorpdfstring{$n$}{n}}
\label{sec:example}

\subsubsection*{For \texorpdfstring{$n=2$}{n=2} qubits}

For $n=2$, the characteristic polynomial
\begin{align*}
\chi_2 = (x+1)^2+x = x^2+x+1
\end{align*}
is irreducible over $\GF_2$. The residue class $\bar{x}$ is a nontrivial element of $\GF_4^\times$, hence has order $3$.

We can easily describe the action of $\Ad_\U$ on the initial generators $\{\Y_1,\Y_2\}$:
\begin{align*}
\Y_1 \mapsto \Z_1\Y_2 \mapsto \X_1\Y_2 \mapsto \Y_1,
\\
\Y_2 \mapsto \Y_1\X_2 \mapsto \Y_1\Z_2 \mapsto \Y_2.
\end{align*}
The complete set of generators of the DLA is therefore
\begin{align*} 
\G_{\infty} = \left\{ \Y_1,\, \Y_2,\, \Z_1\Y_2,\, \Y_1\X_2,\, \X_1\Y_2,\, \Y_1\Z_2 \right\}. 
\end{align*}

The linear span of $i\G_{\infty}$ is already stable by the Lie bracket, so the DLA $\mathfrak g$ is $6$-dimensional. Actually it is $\mathfrak{so}(4)$.

\subsubsection*{For \texorpdfstring{$n=3$}{n=3} qubits}

For $n=3$, the characteristic polynomial
\begin{align*}
\chi_3 = (x+1)^3+x = x^3+x^2+1
\end{align*}
is irreducible over $\GF_2$. The residue class $\bar{x}$ is a nontrivial element of $\GF_8^\times$, hence has order $7$.
Here are the explicit $\Ad_\U$-orbits of the initial generators, which are all of length $7$: 
\[
\begin{array}{r*{14}{@{\,}c@{\,}}l}
\Y_1 &\mapsto& \Y_1\X_2\Z_3 
&\mapsto& \Z_2\Y_3 
&\mapsto& \X_1\Y_3
&\mapsto& \Z_1\Y_2\X_3
&\mapsto& \Y_2
&\mapsto& -\Y_1\Y_2\Y_3
&\mapsto& \Y_1
,\\
\Y_2 &\mapsto& -\Y_1\Y_2\Y_3 
&\mapsto& \Y_1 
&\mapsto& \Y_1\X_2\Z_3
&\mapsto& \Z_2\Y_3
&\mapsto& \X_1\Y_3
&\mapsto& \Z_1\Y_2\X_3
&\mapsto& \Y_2
,\\
\Y_3 &\mapsto& \Y_1\Z_2\X_3 
&\mapsto& \Y_2\X_3 
&\mapsto& \Y_1\X_2\Z_3
&\mapsto& \Y_1\X_2\X_3
&\mapsto& \Y_1\Z_3
&\mapsto& \X_1\Y_2\Z_3
&\mapsto& \Y_3
.
\end{array}
\]

The full DLA $\mathfrak g$ is $\mathfrak{so}(8)$, which is $28$-dimensional.

\subsubsection*{Up to \texorpdfstring{$n=8$}{n=3} qubits}

Table \ref{tab:orders} presents the values of $T_n$ from $n=2$ to $n=8$ qubits, showing different cases where the characteristic polynomial may or may not be irreducible, and the irreducible factors may or may not be primitive. Besides, we indicate the Lie closure depth $L_{\mathrm{Lie}}$ which is actually much smaller; actually it seems (and it should not be too difficult to prove) that $L_{\mathrm{Lie}}$ is $4$ if $n$ is a multiple of $3$, and $2$ otherwise. Finally, we have checked up to $n=8$ that the DLA $\mathfrak g_n$ is $\mathfrak{so}(2^n)$, of dimension $2^{n-1}(2^n - 1)$. It is very likely to be true for all $n$.

\begin{table}[ht]
\caption{\textbf{Orders of the conjugation action and Lie algebra properties as a function of the number of qubits.} For each number of qubits $n$, the table gives the factorization of the characteristic polynomial $\chi_n$ over $\mathbb{F}_2$, the order $T_n$ of the conjugation action, the minimum number of layers $L_{\mathrm{Lie}}$ required to generate the full dynamical Lie algebra $\mathfrak{g}_n$, and its dimension $\dim\mathfrak{g}_n$. Here, $\mathrm{lcm}$ denotes the least common multiple.}
\label{tab:orders}
\centering
\begin{tabular}{|c|c|c|c|c|}
\hline
$n$ & Factorization of $\chi_n$ & $T_n$ & $L_{\mathrm{Lie}}$ &
$\dim\mathfrak g_n$\\
\hline
$2$
& $x^2+x+1$
& $3$ 
& 2
& 6
\\[1mm]

$3$
& $x^3+x^2+1$
& $7$ 
& 4
& 28
\\[1mm]

$4$
& $x^4+x+1$
& $15$ 
& 2
& 120
\\[1mm]

$5$
& $(x^2+x+1)(x^3+x+1)$
& $21 = \lcm(3,7)$ 
& 2
& 496
\\[1mm]

$6$
& $x^6+x^4+x^2+x+1$
& $21$ 
& 4
& 2016
\\[1mm]

$7$
& $x^7+x^6+x^5+x^4+x^3+x^2+1$
& $127$ 
& 2
& 8128
\\[1mm]

$8$
& $(x^2+x+1)(x^6+x^5+x^3+x^2+1)$
& $63 = \lcm(3,63)$ 
& 2
& 32640
\\
\hline
\end{tabular}
\end{table}

\subsection{Implications for overparametrization and barren plateaus}

The exponential growth of the the DLA dimension has implications for overparametrization and Barren Plateaus (BPs). Indeed, as discussed earlier, a necessary condition for overparametrization is $KL \geq \dim \g$. Since each layer contains $K=n$ parameters, this gives a lower bound on the required depth,
\begin{equation}
L\geq\frac{2^{n-1}(2^n-1)}{n}.
\end{equation}
This corresponds to at least $3$ layers for $n=2$, $10$ for $n=3$, $30$ for $n=4$ and $4080$ for $n=8$, illustrating how rapidly the required depth increases with the number of qubits.

Moreover, as shown by Ragone \textit{et al.} \cite{Ragone_2024}, the dimenion of the DLA also plays a role in the emergence of BPs. In fact, under the appropriate assumptions, the variance of the cost function is inversely proportional to $\dim \g$. Since $\dim \g = \mathcal{O}(4^n)$, our Ansatz is therefore susceptible to exhibit BPs in the sufficiently deep regime.

\section*{Conclusion and Perspectives}
\label{sec:con}

In this work, we considered parametrized quantum circuits in which trainable gates are interleaved with the same fixed unitary block $\U$. The main issue is that fixed gates cannot be included directly in the generator set used to construct the dynamical Lie algebra, since they are not continuously parameterized. However, they also cannot be ignored, because they transform the initial generators through conjugation. Our goal was therefore to give a description of the effective generators produced by these fixed blocks.

We first rewrote the circuit in terms of successive conjugations of the original trainable generators. We then focused on the case where the fixed part $\U$ is Clifford and the generators are Pauli strings. In this setting, the conjugated generators remain Pauli strings and follow a finite recurrence up to a sign. Using the binary representation of the projective Pauli group, we showed that Clifford conjugation is represented by a symplectic matrix
$M_{\U}$ over $\GF_2$, and that
\begin{equation}
T_{\U}=\ord(M_{\U})
\end{equation}
gives the global conjugation period. We also derived simpler forms of this matrix when $\U$ is composed only of controlled-$X$ gates or only of controlled-$Z$ gates. These results show how the DLA can be determined for parametrized Ans\"atze built from trainable Pauli rotations and a repeated fixed Clifford block. Instead of adding the fixed Clifford gates to the generator set, one has to consider the successive conjugations of the trainable generators by $\U$. Since these conjugations are periodic, only a finite number of effective generators will be obtained, and the DLA can then be generated by taking the Lie closure of this finite set. We applied these results to an Ansatz in which the fixed gates were composed of $\CX$ gates.

As a perspective, it would be interesting to benchmark two circuits with the same parametrized part but different fixed Clifford blocks, one composed of controlled-$\X$ gates and the other of controlled-$\Z$ gates. This would allow us to study how the choice of the entangling block affects the effective generators and the resulting dynamical Lie algebra.

It would also be interesting to extend our approach to circuits with non-Clifford fixed gates.

\printbibliography

@article{g4zk-79l9,
  title = {Erratum: Provably Trainable Rotationally Equivariant Quantum Machine Learning [PRX Quantum 5, 030320 (2024)]},
  author = {West, Maxwell T. and Heredge, Jamie and Sevior, Martin and Usman, Muhammad},
  journal = {PRX Quantum},
  volume = {6},
  issue = {2},
  pages = {020902},
  numpages = {1},
  year = {2025},
  month = {6},
  publisher = {American Physical Society},
  doi = {10.1103/g4zk-79l9},
  url = {https://link.aps.org/doi/10.1103/g4zk-79l9}
}

@article{Dehaene_2003,
   title={Clifford group, stabilizer states, and linear and quadratic operations over GF(2)},
   volume={68},
   ISSN={1094-1622},
   url={http://dx.doi.org/10.1103/PhysRevA.68.042318},
   DOI={10.1103/physreva.68.042318},
   number={4},
   journal={Physical Review A},
   publisher={American Physical Society (APS)},
   author={Dehaene, Jeroen and De Moor, Bart},
   year={2003},
   month=Oct }

@article{Ragone_2024,
   title={A {L}ie algebraic theory of barren plateaus for deep parameterized quantum circuits},
   volume={15},
   ISSN={2041-1723},
   url={http://dx.doi.org/10.1038/s41467-024-49909-3},
   DOI={10.1038/s41467-024-49909-3},
   number={1},
   journal={Nature Communications},
   publisher={Springer Science and Business Media LLC},
   author={Ragone, Michael and Bakalov, Bojko N. and Sauvage, Frédéric and Kemper, Alexander F. and Ortiz Marrero, Carlos and Larocca, Martín and Cerezo, M.},
   year={2024},
   month=aug }

@article{Peruzzo_2014,
   title={A variational eigenvalue solver on a photonic quantum processor},
   volume={5},
   ISSN={2041-1723},
   url={http://dx.doi.org/10.1038/ncomms5213},
   DOI={10.1038/ncomms5213},
   number={1},
   journal={Nature Communications},
   publisher={Springer Science and Business Media LLC},
   author={Peruzzo, Alberto and McClean, Jarrod and Shadbolt, Peter and Yung, Man-Hong and Zhou, Xiao-Qi and Love, Peter J. and Aspuru-Guzik, Alán and O’Brien, Jeremy L.},
   year={2014},
   month=jul }

@article{Cerezo_2021,
   title={Variational quantum algorithms},
   volume={3},
   ISSN={2522-5820},
   url={http://dx.doi.org/10.1038/s42254-021-00348-9},
   DOI={10.1038/s42254-021-00348-9},
   number={9},
   journal={Nature Reviews Physics},
   publisher={Springer Science and Business Media LLC},
   author={Cerezo, M. and Arrasmith, Andrew and Babbush, Ryan and Benjamin, Simon C. and Endo, Suguru and Fujii, Keisuke and McClean, Jarrod R. and Mitarai, Kosuke and Yuan, Xiao and Cincio, Lukasz and Coles, Patrick J.},
   year={2021},
   month=aug, pages={625–644} }

@article{Larocca_2022,
   title={Diagnosing Barren Plateaus with Tools from Quantum Optimal Control},
   volume={6},
   ISSN={2521-327X},
   url={http://dx.doi.org/10.22331/q-2022-09-29-824},
   DOI={10.22331/q-2022-09-29-824},
   journal={Quantum},
   publisher={Verein zur Forderung des Open Access Publizierens in den Quantenwissenschaften},
   author={Larocca, Martin and Czarnik, Piotr and Sharma, Kunal and Muraleedharan, Gopikrishnan and Coles, Patrick J. and Cerezo, M.},
   year={2022},
   month=sep, pages={824} }

@article{Fontana_2024,
   title={Characterizing barren plateaus in quantum ansätze with the adjoint representation},
   volume={15},
   ISSN={2041-1723},
   url={http://dx.doi.org/10.1038/s41467-024-49910-w},
   DOI={10.1038/s41467-024-49910-w},
   number={1},
   journal={Nature Communications},
   publisher={Springer Science and Business Media LLC},
   author={Fontana, Enrico and Herman, Dylan and Chakrabarti, Shouvanik and Kumar, Niraj and Yalovetzky, Romina and Heredge, Jamie and Sureshbabu, Shree Hari and Pistoia, Marco},
   year={2024},
   month=aug }

@article{West_2024,
   title={Provably Trainable Rotationally Equivariant Quantum Machine Learning},
   volume={5},
   ISSN={2691-3399},
   url={http://dx.doi.org/10.1103/PRXQuantum.5.030320},
   DOI={10.1103/prxquantum.5.030320},
   number={3},
   journal={PRX Quantum},
   publisher={American Physical Society (APS)},
   author={West, Maxwell T. and Heredge, Jamie and Sevior, Martin and Usman, Muhammad},
   year={2024},
   month=jul }

@misc{xiao2025commentprovablytrainablerotationally,
      title={{Comment on ``Provably Trainable Rotationally Equivariant Quantum Machine Learning''}}, 
      author={Zhiming Xiao and Ting Li},
      year={2025},
      eprint={2504.16950},
      archivePrefix={arXiv},
      primaryClass={quant-ph},
      url={https://arxiv.org/abs/2504.16950}, 
}

@misc{gottesman1998heisenbergrepresentationquantumcomputers,
      title={The Heisenberg Representation of Quantum Computers}, 
      author={Daniel Gottesman},
      year={1998},
      eprint={quant-ph/9807006},
      archivePrefix={arXiv},
      primaryClass={quant-ph},
      url={https://arxiv.org/abs/quant-ph/9807006}, 
}

@article{Kandala_2017,
   title={Hardware-efficient variational quantum eigensolver for small molecules and quantum magnets},
   volume={549},
   ISSN={1476-4687},
   url={http://dx.doi.org/10.1038/nature23879},
   DOI={10.1038/nature23879},
   number={7671},
   journal={Nature},
   publisher={Springer Science and Business Media LLC},
   author={Kandala, Abhinav and Mezzacapo, Antonio and Temme, Kristan and Takita, Maika and Brink, Markus and Chow, Jerry M. and Gambetta, Jay M.},
   year={2017},
   month=sep, pages={242–246} }

@misc{iwakiri2025universalityhardwareefficientansatzes,
      title={On universality of hardware-efficient ansatzes}, 
      author={Hokuto Iwakiri and Keita Kanno},
      year={2025},
      eprint={2511.03870},
      archivePrefix={arXiv},
      primaryClass={quant-ph},
      url={https://arxiv.org/abs/2511.03870}, 
}

@article{Larocca_2023,
   title={Theory of overparametrization in quantum neural networks},
   volume={3},
   ISSN={2662-8457},
   url={http://dx.doi.org/10.1038/s43588-023-00467-6},
   DOI={10.1038/s43588-023-00467-6},
   number={6},
   journal={Nature Computational Science},
   publisher={Springer Science and Business Media LLC},
   author={Larocca, Martín and Ju, Nathan and García-Martín, Diego and Coles, Patrick J. and Cerezo, Marco},
   year={2023},
   month=jun, pages={542–551} }

@article{Biamonte_2017,
   title={Quantum machine learning},
   volume={549},
   ISSN={1476-4687},
   url={http://dx.doi.org/10.1038/nature23474},
   DOI={10.1038/nature23474},
   number={7671},
   journal={Nature},
   publisher={Springer Science and Business Media LLC},
   author={Biamonte, Jacob and Wittek, Peter and Pancotti, Nicola and Rebentrost, Patrick and Wiebe, Nathan and Lloyd, Seth},
   year={2017},
   month=sep, pages={195–202} }

@article{Hashagen_2018,
   title={Real Randomized Benchmarking},
   volume={2},
   ISSN={2521-327X},
   url={http://dx.doi.org/10.22331/q-2018-08-22-85},
   DOI={10.22331/q-2018-08-22-85},
   journal={Quantum},
   publisher={Verein zur Forderung des Open Access Publizierens in den Quantenwissenschaften},
   author={Hashagen, A. K. and Flammia, S. T. and Gross, D. and Wallman, J. J.},
   year={2018},
   month=Aug, pages={85} }

@article{Whitfield_2011,
   title={Simulation of electronic structure Hamiltonians using quantum computers},
   volume={109},
   ISSN={1362-3028},
   url={http://dx.doi.org/10.1080/00268976.2011.552441},
   DOI={10.1080/00268976.2011.552441},
   number={5},
   journal={Molecular Physics},
   publisher={Informa UK Limited},
   author={Whitfield, James D. and Biamonte, Jacob and Aspuru-Guzik, Alán},
   year={2011},
   month=mar, pages={735–750} }

@article{Garc_a_lvarez_2020,
   title={Efficient simulatability of continuous-variable circuits with large Wigner negativity},
   volume={2},
   ISSN={2643-1564},
   url={http://dx.doi.org/10.1103/PhysRevResearch.2.043322},
   DOI={10.1103/physrevresearch.2.043322},
   number={4},
   journal={Physical Review Research},
   publisher={American Physical Society (APS)},
   author={García-Álvarez, Laura and Calcluth, Cameron and Ferraro, Alessandro and Ferrini, Giulia},
   year={2020},
   month=Dec}

@article{Helsen_2018,
   title={Representations of the multi-qubit Clifford group},
   volume={59},
   ISSN={1089-7658},
   url={http://dx.doi.org/10.1063/1.4997688},
   DOI={10.1063/1.4997688},
   number={7},
   journal={Journal of Mathematical Physics},
   publisher={AIP Publishing},
   author={Helsen, Jonas and Wallman, Joel J. and Wehner, Stephanie},
   year={2018},
   month=Jul }

@article{Hostens_2005,
   title={Stabilizer states and Clifford operations for systems of arbitrary dimensions and modular arithmetic},
   volume={71},
   ISSN={1094-1622},
   url={http://dx.doi.org/10.1103/PhysRevA.71.042315},
   DOI={10.1103/physreva.71.042315},
   number={4},
   journal={Physical Review A},
   publisher={American Physical Society (APS)},
   author={Hostens, Erik and Dehaene, Jeroen and De Moor, Bart},
   year={2005},
   month=Apr }

@article{Bittel_2026,
   title={A complete theory of the Clifford commutant},
   volume={10},
   ISSN={2521-327X},
   url={http://dx.doi.org/10.22331/q-2026-07-22-2171},
   DOI={10.22331/q-2026-07-22-2171},
   journal={Quantum},
   publisher={Verein zur Forderung des Open Access Publizierens in den Quantenwissenschaften},
   author={Bittel, Lennart and Eisert, Jens and Leone, Lorenzo and Mele, Antonio A. and Oliviero, Salvatore F.E.},
   year={2026},
   month=Jul, pages={2171} }

@misc{Mansky2023bsb,
      title={Decomposition Algorithm of an Arbitrary Pauli Exponential through a Quantum Circuit}, 
      author={Maximilian Balthasar Mansky and Victor Ramos Puigvert and Santiago Londoño Castillo and Claudia Linnhoff-Popien},
      year={2023},
      eprint={2305.04807},
      archivePrefix={arXiv},
      primaryClass={quant-ph},
      url={https://arxiv.org/abs/2305.04807}, 
}

\end{document}